\documentclass{article}

\usepackage{times}
\usepackage{soul}
\usepackage{url}
\usepackage[hidelinks]{hyperref}
\usepackage[utf8]{inputenc}
\usepackage[small]{caption}
\usepackage{graphicx}
\usepackage{amsmath}
\usepackage{amsthm}
\usepackage{booktabs}
\usepackage{tikz}
\usepackage{tikz-3dplot}
\usepackage{comment}

\usetikzlibrary{3d,calc}
\usepackage{xcolor}

\usepackage{algorithm}
\usepackage{algorithmicx}
\usepackage{algpseudocode}

\algnewcommand{\algorithmicforeach}{\textbf{foreach}}
\algdef{SE}[FOR]{ForEach}{EndForEach}[1]{\algorithmicforeach\ #1\ \algorithmicdo}{\algorithmicend\ \algorithmicforeach}

\usepackage[switch]{lineno}

\usepackage{amsthm, amssymb}
\usepackage{thmtools}

\usepackage{mathtools}
\usepackage{cleveref}

\newcommand{\cqed}{\renewcommand{\qedsymbol}{$\lrcorner$}\qed}

\theoremstyle{definition}

\newtheorem{theorem}{Theorem}
\usepackage{amsmath}
\usepackage{amsthm}
\usepackage{amsfonts} 

\usepackage{tikz}
\usetikzlibrary{intersections}
\usetikzlibrary{patterns}

\usepackage{todonotes}
\usepackage{xspace}
\newtheorem{lemma}[theorem]{Lemma}

\newtheorem{definition}{Definition}

\newcommand{\Oh}{\mathcal{O}}
\DeclareMathOperator{\poly}{poly}

\newcommand{\EnumerateSeparator}{\texttt{EnumerateSeparator}\xspace}
\newcommand{\Divide}{\texttt{Divide}\xspace}
\newcommand{\ExhaustiveSearch}{\texttt{ExhaustiveSearch}\xspace}
\newcommand{\ConstructEF}{\texttt{ConstructEF}\xspace}
\newcommand{\CheckAllocation}{\texttt{CheckAllocation}\xspace}
\newcommand{\Yield}{\textbf{yield}\xspace}
\newcommand{\None}{\ensuremath{\bot}}
\newcommand{\False}{\texttt{False}\xspace}
\newcommand{\True}{\texttt{True}\xspace}

\newcommand{\EF}{\textsc{EF}\xspace}
\newcommand{\EFone}{\textsc{EF1}\xspace}

\newcommand{\EFX}{\textsc{EFX}\xspace}

\newcommand{\MNW}{\textsc{MNW}\xspace}

\DeclareMathOperator{\operatorClassNP}{{\sf NP}}
\newcommand{\classNP}{\ensuremath{\operatorClassNP}}

\usepackage{PRIMEarxiv}

\usepackage[utf8]{inputenc} 
\usepackage[T1]{fontenc}    
\usepackage{hyperref}       
\usepackage{url}            
\usepackage{booktabs}       
\usepackage{amsfonts}       
\usepackage{nicefrac}       
\usepackage{microtype}      
\usepackage{lipsum}
\usepackage{fancyhdr}       
\usepackage{graphicx}       
\graphicspath{{media/}}     

\title{Subexponential Algorithm for High Multiplicity Fair Division of Mixed Instances via Stereometry}

\author{
  Yuriy Dementiev\\
  ITMO University
  \And
  Fedor Pribytkov\\
  St. Petersburg State University\\
  \And
  Danil Sagunov\\
  ITMO University\\
}

\begin{document}
\maketitle
\begin{abstract}
We study the problem of computing an envy-free~(\EF) allocation of $m$ indivisible items among $n$ agents when items come in three distinct types.
Each agent holds additive valuations over item types that may be positive (goods), negative (chores), or mixed. 
We present the first subexponential-time algorithm with running time time $(n \cdot m)^{\Oh(\sqrt{n})}$ that finds an \EF allocation whenever one exists, or correctly reports that none exists. 
Our approach exploits a geometric representation of \EF allocations as convex polyhedra in $\mathbb{R}^3$ and applies Miller's planar cycle-separator theorem to recursively decompose the agent set into balanced subgroups. 
We further extend the algorithm to handle agents whose allocations are fixed in advance, preserving envy-freeness across all agents. 


\end{abstract}

\section{Introduction}
Fair division of resources among self-interested agents is a foundational
problem in artificial intelligence and multi-agent
systems~\cite{moulin2004fair,brandt2016handbook}. Among fairness criteria,
\emph{envy-freeness}~(EF) is one of the most compelling: an allocation is
envy-free if every agent weakly prefers their own bundle to that of any other
agent. While EF allocations always exist for divisible
goods~\cite{steinhaus1948problem}, the indivisible setting is significantly more
challenging---an EF allocation may not exist at all, and deciding whether one
exists is computationally demanding in general.

Most work on envy-free allocation of indivisible goods either imposes special
structural assumptions (e.g., identical valuations, binary preferences, two
agents) or relaxes the fairness requirement to approximation notions such as
\EFone~\cite{lipton2004approximately,budish2011combinatorial}
(envy-freeness up to one item) or
\EFX~\cite{caragiannis2019unreasonable,EFX_PO}. Exact \EF for general additive
valuations over indivisible goods is known to be computationally
hard~\cite{bouveret2008efficiency}, and practical algorithms typically resort
to exhaustive search, which is exponential in the number of agents.

In this paper we study a natural and practically motivated special case: items
come in $t$ types, and each agent has additive valuations that depend only on
the \emph{type} of an item rather than on individual items. This captures
scenarios such as computing resource allocation across hardware tiers, food
distribution by nutritional category, or budget allocation across departments.
Crucially, our setting encompasses not only \emph{goods} (positive valuations)
but also \emph{chores} (negative valuations) and \emph{mixed instances} in
which some item types are goods for certain agents and chores for others. The only constraint is we forbid agents with zero valuation vectors.
The EF condition is uniform across all cases: every agent must weakly prefer
their own bundle over anyone else's. Despite the restriction to a fixed
number of types, the problem remains non-trivial: the space of feasible
allocations has size $((m_1+1) \cdot\ldots \cdot (m_t+1))^n$, where $m_i$ is the total number of items of type $i$, rendering brute-force search
infeasible for large inputs.

The number of types $t$ turns out to be the key parameter governing
tractability, and the known results reveal a striking complexity landscape.

For \emph{two types} ($t = 2$), was shown that an EF allocation can be found in polynomial time.
For \emph{many types} ($t = \Omega(m)$), the problem is hard even for two agents: exact \EF with two agents reduces to a variant of \textsc{Subset Sum}, and an ETH-based lower bound~\cite{cygan2015parameterized} implies that no $2^{o(m)}$-time algorithm exists. In particular, the gap between $t=2$ and $t=\Omega(m)$ cannot be closed without a breakthrough on ETH-hard problems.

The case $t = 3$ studied in this paper lies between these two extremes and requires fundamentally new techniques. For the special case of a constant number of types, Maximum Nash Welfare allocations can be computed in polynomial time~\cite{nguyen2014mnw}; however, this does not yield exact \EF algorithms, since \MNW and \EF are inequivalent for indivisible items.

\paragraph{Our contribution.}
We present the first subexponential algorithm for exact envy-free allocation in the three-type setting, covering goods, chores, and mixed instances.
The key insight is geometric: each agent's valuation vector
$v_i \in \mathbb{R}^3$ defines a normal to a hyperplane, and a feasible allocation can be represented as a point in a convex polyhedron in $\mathbb{R}^3$. The \EF conditions
\[
   \forall\, i,j \in N: \qquad \langle v_i,\, x_i - x_j \rangle \;\geq\; 0
\]
translate into a system of linear half-space constraints, so the set of all \EF allocations forms a bounded convex polytope. We exploit this structure through two novel procedures:
\texttt{EnumerateSeparator}: finds a small cycle separator of size at most $\sqrt{8n}$ on the planar graph induced by the EF polytope; its existence is guaranteed by Miller's classical planar cycle-separator theorem~\cite{miller1986finding}.
\texttt{Divide}: uses the spherical image of the separator to
partition the remaining agents into two balanced groups $L$ and $R$, each of size at most $\tfrac{11}{12}n$, such that any \EF solutions for $L$ and $R$ can be merged into an \EF solution for the whole.

Combining these in a divide-and-conquer scheme gives us the
\texttt{ConstructEF} algorithm with time complexity $(n \cdot m)^{\Oh(\sqrt{n})}.$
We also consider an extended setting with \emph{fixed agents}---agents whose allocations are predetermined as part of the input---and show that the algorithm extends naturally to this case.

\paragraph{Related work.}
Envy-free allocation of indivisible goods has been studied extensively in the computational social choice community. Exact \EF algorithms are known for two agents~\cite{bliem2016complexity} and for special valuation structures~\cite{aziz2015fair}. For the general multi-agent case, the problem is \classNP-hard when an \EF allocation may not exist~\cite{bouveret2008efficiency}.
The multi-type (also called \emph{multi-dimensional} or \emph{categorised}) setting was studied in~\cite{mackin2016allocating,10.5555/3463952.3463988} but without
subexponential guarantees. \EF for chores and mixed manna has attracted significant recent attention~\cite{bhaskar2021chores,aziz2021mixedmanna,bogomolnaia2017dividing}; to the best of our knowledge, no subexponential exact algorithm was previously known for any non-trivial multi-agent, multi-type variant of this problem, regardless of the sign of valuations.

\paragraph{Organization of the paper.}
Section~\ref{sec:prelim} introduces the formal model and notation.
\Cref{sec:algorithm} presents our main algorithmic result (\Cref{thm:algorithm}) in the form of a recursive subexponential algorithm, explains the idea behind it, proves its correctness and running time bound.
Subsequent sections (\Cref{sec:enumerate}, \Cref{sec:divide} and \Cref{sec:merge}) are devoted to the proof of the main structural result (\Cref{thm:main}) that our algorithm relies on.
\section{Preliminaries and Notation}
\label{sec:prelim}

We consider a fair division problem with a finite set of agents and multiple types of indivisible goods.

\paragraph{Agents and goods.}
Let $N = \{1, \dots, n\}=[n]$ denote the set of all agents,  
$M$ denote the set of all items and $m = |M|$ is the total number of items. 
Let $T = \{1,2,3\}$ denote the set of item types.  
For each type $t \in T$, let $m_t \in \mathbb{Z}_{\ge 0}$ be the total number of available items of type $t$: $m_1+m_2+m_3 = m$.
In our algorithms, we will deal with subsets of goods, characterized by three non-negative integers $k_i\le m_i$, for each $i\in T$.
We denote $\mathbf{k} = (k_1, k_2, k_3)$.

\paragraph{Valuations.}
Each agent $i \in N$ has additive linear valuations over item types, represented by a vector
\[
{v}_i = (v_i(1), v_i(2), v_i(3)) \in \mathbb{R}_{\ge 0}^3.
\]
The only restriction we put on a vector $v_i$ for each $i$ is that  $v_i$ is non-degenerate, that is, $v_i\neq(0,0,0)$.
Given an allocation $x_i \in \mathbb{Z}_{\ge 0}^3$, the utility of agent $i$ is defined as
\[
u_i(x_i) = \sum_{t \in T} v_i(t)\cdot x_{i,t} = \langle v_i, x_i \rangle.
\]

\paragraph{Allocations.}
Let $A \subseteq N$ be set of agents. An allocation for $A$ is a collection $X = (x_i)_{i \in A}$, where each
\[
x_i = (x_{i,1}, x_{i,2}, x_{i,3}) \in \mathbb{Z}_{\ge 0}^3
\]
specifies the number of items of each type assigned to agent $i$.
An allocation is \emph{feasible} for $\mathbf{k}$ if it satisfies the resource constraints,
$\sum_{i \in A} x_{i,t} = k_t$ for each $t \in T$.

\paragraph{Envy-freeness.}
An allocation $X = (x_i)_{i \in A}$ is called \emph{envy-free} (\EF) if for any pair of agents $i,j \in A$,
$u_i(x_i) \ge u_i(x_j)$.

\paragraph{Fixed agents.}
In the extended setting, we are also given a subset of agents $F \subseteq A$ whose allocations are fixed in advance.  
For each $i \in F$, the vector $x_i \in \mathbb{Z}_{\ge 0}^3$ is specified as a part of the input.

We denote the remaining agents by $A_{\text{free}} = A \setminus F$.  
The goal is to compute allocations $x_i$ for all $i \in A_{\text{free}}$ such that the combined allocation over all agents $A$ is feasible and envy-free.

We say that two partial allocations \emph{agree}, if for each agent that is given a bundle in both allocations, this bundle is the same in both allocations for this agent.

\paragraph{Geometric interpretation.}
We view each allocation vector $x_i \in \mathbb{R}^3$ as a point in $\mathbb{R}^3$, and each valuation vector $v_i$ as a normal vector.  
The envy-freeness condition between agents $i$ and $j$ with bundles $x_i$ and $x_j$ respectively can be written as
\[
\langle v_i, x_i - x_j \rangle \ge 0.
\]
That is, $x_i$ and $v_i$ define a half-space for envy-free bundles $x_j$ in $\mathbb{R}^3$.  
Thus, the set of envy-free allocations can be interpreted as a region defined by linear constraints in $\mathbb{R}^3$. This is also the reason why we forbid degenerate valuations: such valuations do not define half-spaces.  

\begin{figure}
\begin{tabular}{lr}
\hspace*{-50px}
    \tdplotsetmaincoords{70}{115}
\definecolor{faceCut0}{rgb}{0.224,0.224,0.575}
\definecolor{faceCut1}{rgb}{0.274,0.274,0.705}
\definecolor{faceCut2}{rgb}{0.299,0.299,0.770}
\definecolor{faceCut3}{rgb}{0.122,0.122,0.315}
\definecolor{faceCut4}{rgb}{0.325,0.325,0.835}
\definecolor{faceCut5}{rgb}{0.173,0.173,0.445}
\definecolor{faceCut6}{rgb}{0.249,0.249,0.640}
\definecolor{faceCut7}{rgb}{0.148,0.148,0.380}
\definecolor{faceCut8}{rgb}{0.350,0.350,0.900}
\definecolor{faceCut9}{rgb}{0.198,0.198,0.510}
\definecolor{faceCoord0}{rgb}{0.334,0.334,0.334}
\definecolor{faceCoord1}{rgb}{0.193,0.193,0.193}
\definecolor{faceCoord2}{rgb}{0.193,0.193,0.193}

\begin{tikzpicture}[
    tdplot_main_coords, scale=5, >=stealth,
    vertex/.style={circle, fill=black, inner sep=0.5pt},
    axis/.style={->, thick, black},
    edge/.style={draw=blue!40!black!70, line width=0.65pt},
]

\draw[axis] (0,0,0) -- (1.15,0,0) node[above] {$x$};
\draw[axis] (0,0,0) -- (0,1.15,0) node[above] {$y$};
\draw[axis] (0,0,0) -- (0,0,1.15) node[above] {$z$};

\coordinate (V1) at (0.3250,0.6604,0.5388);
\coordinate (V2) at (0.5893,0.6278,0.3071);
\coordinate (V3) at (0.6278,0.3071,0.5893);
\coordinate (V4) at (0.6604,0.5388,0.3250);
\coordinate (V5) at (0.5388,0.3250,0.6604);
\coordinate (V6) at (0.3071,0.5893,0.6278);
\coordinate (V7) at (0.2535,0.0201,0.9344);
\coordinate (V8) at (0.1183,0.5173,0.7444);
\coordinate (V9) at (0.2473,-0.0000,0.9420);
\coordinate (V10) at (-0.0000,0.5392,0.7360);
\coordinate (V11) at (-0.0000,-0.0000,0.9420);
\coordinate (V12) at (0.0201,0.9344,0.2535);
\coordinate (V13) at (0.5173,0.7444,0.1183);
\coordinate (V14) at (-0.0000,0.9420,0.2473);
\coordinate (V15) at (0.5392,0.7360,-0.0000);
\coordinate (V16) at (-0.0000,0.9420,-0.0000);
\coordinate (V17) at (0.7444,0.1183,0.5173);
\coordinate (V18) at (0.9344,0.2535,0.0201);
\coordinate (V19) at (0.7360,-0.0000,0.5392);
\coordinate (V20) at (0.9420,0.2473,-0.0000);
\coordinate (V21) at (0.9420,-0.0000,-0.0000);
\coordinate (V22) at (-0.0000,-0.0000,-0.0000);

\filldraw[fill=faceCoord0, draw=gray!40, opacity=0.35, line width=0.4pt]
    (V16) -- (V14) -- (V10) -- (V11) -- (V22) -- cycle;
\filldraw[fill=faceCoord1, draw=gray!40, opacity=0.35, line width=0.4pt]
    (V22) -- (V11) -- (V9) -- (V19) -- (V21) -- cycle;
\filldraw[fill=faceCoord2, draw=gray!40, opacity=0.35, line width=0.4pt]
    (V21) -- (V20) -- (V15) -- (V16) -- (V22) -- cycle;

\filldraw[fill=faceCut3, draw=blue!40!black!70, opacity=0.55, line width=0.5pt]
    (V20) -- (V18) -- (V17) -- (V19) -- (V21) -- cycle;
\filldraw[fill=faceCut7, draw=blue!40!black!70, opacity=0.55, line width=0.5pt]
    (V19) -- (V17) -- (V3) -- (V5) -- (V7) -- (V9) -- cycle;
\filldraw[fill=faceCut5, draw=blue!40!black!70, opacity=0.55, line width=0.5pt]
    (V4) -- (V3) -- (V17) -- (V18) -- cycle;
\filldraw[fill=faceCut9, draw=blue!40!black!70, opacity=0.55, line width=0.5pt]
    (V15) -- (V13) -- (V2) -- (V4) -- (V18) -- (V20) -- cycle;
\filldraw[fill=faceCut0, draw=blue!40!black!70, opacity=0.55, line width=0.5pt]
    (V4) -- (V2) -- (V1) -- (V6) -- (V5) -- (V3) -- cycle;
\filldraw[fill=faceCut6, draw=blue!40!black!70, opacity=0.55, line width=0.5pt]
    (V7) -- (V5) -- (V6) -- (V8) -- cycle;
\filldraw[fill=faceCut1, draw=blue!40!black!70, opacity=0.55, line width=0.5pt]
    (V11) -- (V9) -- (V7) -- (V8) -- (V10) -- cycle;
\filldraw[fill=faceCut2, draw=blue!40!black!70, opacity=0.55, line width=0.5pt]
    (V16) -- (V14) -- (V12) -- (V13) -- (V15) -- cycle;
\filldraw[fill=faceCut4, draw=blue!40!black!70, opacity=0.55, line width=0.5pt]
    (V12) -- (V1) -- (V2) -- (V13) -- cycle;
\filldraw[fill=faceCut8, draw=blue!40!black!70, opacity=0.55, line width=0.5pt]
    (V10) -- (V8) -- (V6) -- (V1) -- (V12) -- (V14) -- cycle;

\node[vertex, fill=black] at (V1) {};
\node[vertex, fill=black] at (V2) {};
\node[vertex, fill=black] at (V3) {};
\node[vertex, fill=black] at (V4) {};
\node[vertex, fill=black] at (V5) {};
\node[vertex, fill=black] at (V6) {};
\node[vertex, fill=black] at (V7) {};
\node[vertex, fill=black] at (V8) {};
\node[vertex, fill=black] at (V9) {};
\node[vertex, fill=black] at (V10) {};
\node[vertex, fill=black] at (V11) {};
\node[vertex, fill=black] at (V12) {};
\node[vertex, fill=black] at (V13) {};
\node[vertex, fill=black] at (V14) {};
\node[vertex, fill=black] at (V15) {};
\node[vertex, fill=black] at (V16) {};
\node[vertex, fill=black] at (V17) {};
\node[vertex, fill=black] at (V18) {};
\node[vertex, fill=black] at (V19) {};
\node[vertex, fill=black] at (V20) {};
\node[vertex, fill=black] at (V21) {};
\node[vertex, fill=gray!50] at (V22) {};

\draw[edge] (V1)--(V2);
\draw[edge] (V1)--(V6);
\draw[edge] (V1)--(V12);
\draw[edge] (V2)--(V4);
\draw[edge] (V2)--(V13);
\draw[edge] (V3)--(V4);
\draw[edge] (V3)--(V5);
\draw[edge] (V3)--(V17);
\draw[edge] (V4)--(V18);
\draw[edge] (V5)--(V6);
\draw[edge] (V5)--(V7);
\draw[edge] (V6)--(V8);
\draw[edge] (V7)--(V8);
\draw[edge] (V7)--(V9);
\draw[edge] (V8)--(V10);
\draw[edge] (V9)--(V11);
\draw[edge] (V9)--(V19);
\draw[edge] (V10)--(V11);
\draw[edge] (V10)--(V14);
\draw[edge] (V11)--(V22);
\draw[edge] (V12)--(V13);
\draw[edge] (V12)--(V14);
\draw[edge] (V13)--(V15);
\draw[edge] (V14)--(V16);
\draw[edge] (V15)--(V16);
\draw[edge] (V15)--(V20);
\draw[edge] (V16)--(V22);
\draw[edge] (V17)--(V18);
\draw[edge] (V17)--(V19);
\draw[edge] (V18)--(V20);
\draw[edge] (V19)--(V21);
\draw[edge] (V20)--(V21);
\draw[edge] (V21)--(V22);


\end{tikzpicture} & \begin{tikzpicture}[scale=11, line cap=round, line join=round]
\definecolor{myblue}{rgb}{0.1, 0.3, 0.8}
\definecolor{myred}{rgb}{0.8, 0.1, 0.1}
\definecolor{mygreen}{rgb}{0.1, 0.8, 0.1}

\draw[myblue, line width=0.6pt] (-0.23559, -0.30160) -- (-0.25517, 0.08832);
\draw[myblue, line width=0.6pt] (-0.25517, 0.08832) -- (0.00015, 0.00013);
\draw[myblue, line width=0.6pt] (-0.14467, 0.35372) -- (0.20422, 0.17716);
\draw[myblue, line width=0.6pt] (-0.15348, -0.11368) -- (0.00015, 0.00013);
\draw[myblue, line width=0.6pt] (0.37938, -0.05318) -- (0.17545, -0.07607);
\draw[myblue, line width=0.6pt] (0.05142, -0.26508) -- (0.17545, -0.07607);
\draw[myblue, line width=0.6pt] (-0.25517, 0.08832) -- (-0.15348, -0.11368);
\draw[myblue, line width=0.6pt] (0.05142, -0.26508) -- (0.00015, 0.00013);
\draw[myblue, line width=0.6pt] (-0.14467, 0.35372) -- (-0.02160, 0.19010);
\draw[myblue, line width=0.6pt] (0.37938, -0.05318) -- (0.20422, 0.17716);
\draw[myblue, line width=0.6pt] (-0.14467, 0.35372) -- (-0.25517, 0.08832);
\draw[myblue, line width=0.6pt] (0.20422, 0.17716) -- (-0.02160, 0.19010);
\draw[myblue, line width=0.6pt] (0.37938, -0.05318) -- (0.05142, -0.26508);
\draw[myblue, line width=0.6pt] (-0.02160, 0.19010) -- (0.00015, 0.00013);
\draw[myblue, line width=0.6pt] (-0.25517, 0.08832) -- (-0.02160, 0.19010);
\draw[myblue, line width=0.6pt] (-0.23559, -0.30160) -- (0.05142, -0.26508);
\draw[myblue, line width=0.6pt] (-0.23559, -0.30160) -- (-0.15348, -0.11368);
\draw[myblue, line width=0.6pt] (0.05142, -0.26508) -- (-0.15348, -0.11368);
\draw[myblue, line width=0.6pt] (0.20422, 0.17716) -- (0.00015, 0.00013);
\draw[myblue, line width=0.6pt] (0.20422, 0.17716) -- (0.17545, -0.07607);
\draw[myblue, line width=0.6pt] (0.17545, -0.07607) -- (0.00015, 0.00013);

\tikzset{every node/.style={}}
\fill[myred] (-0.14467, 0.35372) circle (0.2pt);
\fill[myred] (0.37938, -0.05318) circle (0.2pt);
\fill[myred] (0.20422, 0.17716) circle (0.2pt);
\fill[myred] (-0.23559, -0.30160) circle (0.2pt);
\fill[myred] (-0.25517, 0.08832) circle (0.2pt);
\fill[myred] (0.05142, -0.26508) circle (0.2pt);
\fill[myred] (0.17545, -0.07607) circle (0.2pt);
\fill[myred] (-0.02160, 0.19010) circle (0.2pt);
\fill[myred] (-0.15348, -0.11368) circle (0.2pt);
\fill[myred] (0.00015, 0.00013) circle (0.2pt);
\end{tikzpicture} \\
\end{tabular}
\caption{An example of a convex polyhedron (left) given by the intersection of the halfspaces given by the agents' valuations (hyperplane normal vector) and their bundles (point on the hyperplane), and the corresponding planar graph (right), where agents become vertices and share an edge iff their hyperplanes share an edge.
Intuitively, this graph is the ``envy'' graph, suggesting to check envy only between agents that are neighbors in the graph.}\label{fig:3d}
\end{figure}
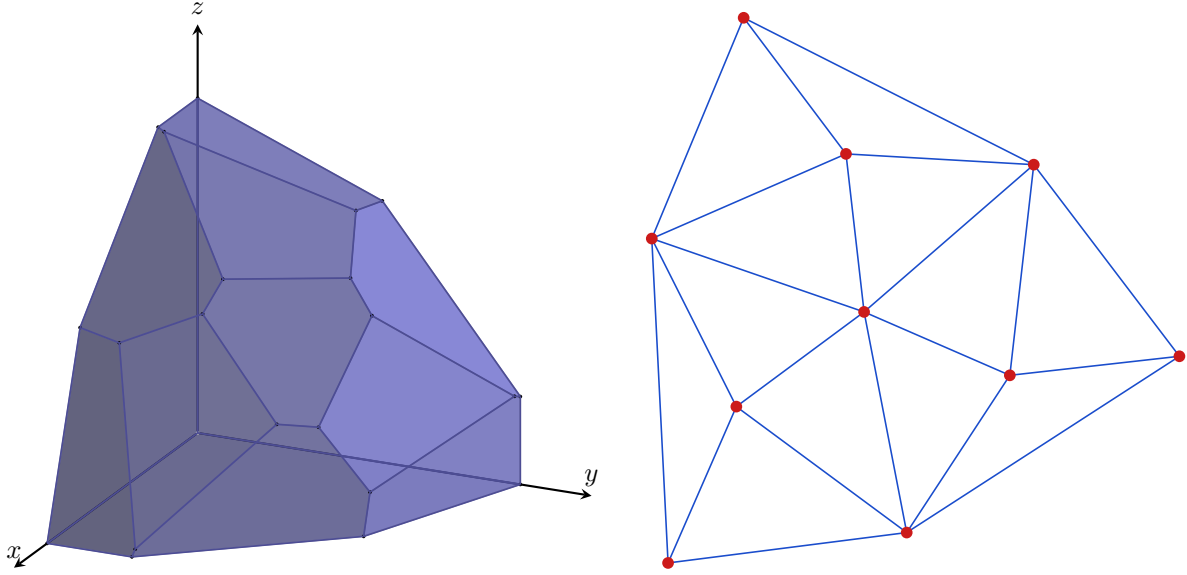
\section{Algorithm}\label{sec:algorithm}
In this section, we describe the recursive algorithm $\ConstructEF$ for finding an envy-free allocation, assuming the existence of procedures $\EnumerateSeparator$ and $\Divide$.
As announced in the introduction, these two procedures allow to fix a certain allocation and split the agents into two groups that are not very much different in size.
We introduce their properties in the following technical result.

\begin{theorem}
\label{thm:main}

There exist efficient algorithms $\Divide$ and $\EnumerateSeparator$ and universal constants
$C_1, C_2$ with $C_2 < 1$ such that the following holds.

Let $A$ be a set of agents, let $\mathbf{k} = (k_1, k_2, k_3)$ be the item counts,
and let $X$ be an arbitrary envy-free allocation for $A$.
Then during its enumeration, $\EnumerateSeparator(A,\,\mathbf{k})$ yields a set of agents $S$ and a partial allocation 
$X_S$ of items to agents in $S$ such that $\Divide(A,\, X_S)$ returns a partition
$A = L \cup R$ satisfying:
\begin{enumerate}
    \item $X$ agrees with $X_S$, that is, $X|_S=X_S$;
    \item $|S| \leq C_1 \cdot \sqrt{|A|}$ ;\label{prop:size}
    \item  $L \cap R = S$ \quad and \quad $|L|,\;|R| \leq C_2\cdot|A|$; \label{prop:balance}
    \item for \emph{any} pair of allocations $X_L$ and $X_R$ that are envy-free for $L$ and $R$ respectively and both agree with $X_S$,
          the combined allocation $X_L \cup X_R$ is envy-free for all of $A$. \label{prop:merge}
\end{enumerate}

\end{theorem}

We will prove \Cref{thm:main} later.
Before that, we use it to prove the main result of this paper and present \texttt{ConstructEF}.

\subsection{Description of \ConstructEF}

\paragraph{High-level idea.}

The algorithm follows a divide-and-conquer paradigm. The key idea is to exploit the geometric structure of envy-free allocations. In particular, we assume that any feasible envy-free allocation induces a geometric object (a convex polyhedron).
Its planar graph admits a small separator of size $\Oh(\sqrt{n})$ (see \Cref{fig:3d}).

The algorithm proceeds as follows:
\begin{itemize}
    \item It tries to guess a small set of agents together with their allocation — $X_S$.
    \item  ``Removing'' $X_S$ separates the remaining agents into the \emph{left} part and the \emph{right} part.
    $L$ and $R$ are formed as sets of the agents of the corresponding parts, together with the separator agent set $S$ (so $L\cap R=S$ and $L\cup R=A$).
    \item The algorithm tries every possible way to split the remaining items between the parts.
    \item For each such guess, the algorithm recursively constructs envy-free allocations on $L$ and $R$, treating $X_S$ as a set of agents with fixed allocations.
    This way, the two independent solutions agree on $S$.
    \item Finally, it combines the solutions if both recursive calls succeed.
\end{itemize}

Since the separator is small, the number of possibilities for  $X_S$ is subexponential.
Since the separator is balanced, that is, splits the agents in not very small parts, the depth of the recursion is logarithmic in the number of agents.
This leads to the overall subexponential running time.

\paragraph{Detailed description.}
\texttt{ConstructEF} is present in the form of pseudocode in \Cref{alg:myalg}.
It exactly implements the high-level idea above.
The only technical addition is the recursion base: when $A$ is small, the algorithm finds a solution by brute force (the \texttt{ExhaustiveSearch} procedure).
Note that the algorithm keeps track of the fixed allocations $X_F$, and these allocations are propagated in recursion and distributed between the two recursion parts together with the agents.
We give description of this and another newly-introduced procedures below.

\begin{algorithm}[ht]
\small
\caption{\ConstructEF$(A,X_F,\mathbf{k})$: a recursive subexponential algorithm that constructs an \EF allocation of a subset of items defined by $\mathbf{k}$ to a set of agents $A$, in a way that extends a partial allocation $X_F$.}
\label{alg:myalg}
\begin{algorithmic}[1]

\Require $A \subseteq [n]$ --- set of agents;
\Require $F \subseteq A$ --- set of agents whose allocations are fixed;
\Require $X_F = (i, (x_{i,1}, x_{i,2}, x_{i,3}))_{i \in F}$ --- fixed allocation for agents in $F$;
\Require $\mathbf{k} = (k_1, k_2, k_3)$ --- total number of items of each type to allocate for $A$.

\Statex

\If{$|A| \leq 1200$}
    \State \Return $\ExhaustiveSearch(A, X_F, \mathbf{k})$;
\EndIf
\Loop
    \State $(S, X_S) \gets \EnumerateSeparator(A, \mathbf{k})$;
    \If{$X_S = \None$}
        \State \Return \None;
    \EndIf
    \If{$X_S$ does not agree with $X_F$}
        \State \textbf{continue};
    \EndIf
    \State $\mathbf{k}^S = (k^S_1, k^S_2, k^S_3) \gets \sum_{X_S}x_i $;
    \ForEach{$\mathbf{k}^L \in \mathbb{Z}_{\ge 0}^3$ s. t. $\mathbf{k}^S\le \mathbf{k}^L \leq \mathbf{k}$}
        \State $(L,\ R) \gets \Divide(A, X_S)$;
        \State $\mathbf{k}^R \gets \mathbf{k} + \mathbf{k}^S - \mathbf{k}^L$;
        \State $\EF_L \gets \ConstructEF\left(L, X_{(F\cap L)} \cup X_{S},  \mathbf{k}^L\right)$;
        \State $\EF_R \gets \ConstructEF\left(R, X_{(F\cap R)} \cup X_{S}, \mathbf{k}^R\right)$;
        \If{$\EF_L \neq \None$ \textbf{and} $\EF_R \neq \None$ \textbf{and} $\texttt{IsEF}(\EF_L \cup \EF_R)$}
            \State \Return $\EF_L \cup \EF_R$;
        \EndIf
    \EndForEach
\EndLoop

\end{algorithmic}
\end{algorithm}

\paragraph{Exhaustive search.}
For completeness, we present the $\ExhaustiveSearch$ procedure in \Cref{alg:brute} in the form of pseudocode.
The idea is to enumerate all feasible allocations of items to agents in $A$, respecting: (i) the fixed allocations $X_F$, (ii) the total item counts $\mathbf{k}$, and check the envy-freeness condition.

\begin{algorithm}[ht]
\small
\caption{\ExhaustiveSearch$(A,X_F,\mathbf{k})$: an algorithm that constructs an envy-free allocation of an item subset given by $\mathbf{k}$ to agents in $A$, in a way that extends $X_F$.}\label{alg:brute}
\begin{algorithmic}[1]
\Require $A \subseteq [n]$ --- set of agents;
\Require $F \subseteq A $ --- set of agents whose allocations are fixed;
\Require $X_F = (i, (x_{i,1}, x_{i,2}, x_{i,3}))_{i \in F}$ --- fixed allocation for agents in $F$;
\Require $\mathbf{k} = (k_1, k_2, k_3)$ --- number of items of each type;
\Statex

\For{each assignment $X_{A\setminus F} = (i, (x_i))_{i \in A \setminus F}$ such that
    \Statex $\forall i \in{A\setminus F}: x_i \in \mathbb{Z}_{\ge 0}^3$ and $\sum_{i \in A} x_{i} = \mathbf{k}$ }
\State $X \gets X_{A\setminus F} \cup X_F$;
\If{\texttt{IsEF}($X$)}
    \State \Return $X$;
\EndIf
\EndFor

\State \Return $\None$;

\end{algorithmic}
\end{algorithm}

The procedure checks all feasible allocations consistent with the resource constraints and fixed agents. Since the search space is finite and all allocations are explicitly tested, it returns an envy-free allocation if one exists, and $\None$ otherwise.
The number of allocations is bounded by $(k_1 \cdot k_2 \cdot k_3)^{|A \setminus F|}$, this determines the running time of the procedure.
It is used by our main algorithm only in case $|A|=\Oh(1)$, so any run of \texttt{ExhaustiveSearch} runs in time polynomial in $m$.
\paragraph{Envy-freeness check.}
\texttt{IsEF} is a polynomial-time procedure that evaluates whether a given (partial) allocation is envy-free and the constraints on the number of items are satisfied for each type.

\subsection{Correctness and running time analysis}
In what remains of this section, we prove that \texttt{ConstructEF} is correct and runs in time subexponential in $n$.
We start with proving the correctness.

\begin{lemma}
Algorithm $\ConstructEF$ returns an EF allocation of items with quantities $\mathbf{k}$ for agents in the set $A$ that extends $X_F$. If no such allocation exists, it returns $\None$.
\end{lemma}
\begin{proof}
We prove correctness by induction on $p = |A|$.

\textbf{Base case.}  
If $p\leq 1200$, the algorithm calls $\ExhaustiveSearch$ and gives its output, that is clearly correct.

\textbf{Inductive step.}  
Assume the statement holds for all sets of agents of size strictly smaller than $p$.
Suppose there exists an envy-free allocation $X$ for the current instance $(A, X_F, \mathbf{k})$.

By \Cref{thm:main}, there exists a separator $S \subseteq A$ of size $\Oh(\sqrt{n})$ together with an allocation $X_S$ such that (i) removing $S$ partitions $A \setminus S$ into two sets $L\setminus S$ and $R\setminus S$; (ii) the restriction of $X$ to $L$ and $R$ remains extendable independently; (iii)
combining $X_S$, $X_L=X|_L$, and $X_R=X|_R$ is an envy-free allocation.
By \Cref{thm:main}, during its execution, $\EnumerateSeparator$ eventually outputs $S$ and $X_S$.

Let $\mathbf{k}^L$ and $\mathbf{k}^R$ be the corresponding vector of item counts assigned to $L$ and $R$ in $X$. Then the algorithm will consider this choice of $\mathbf{k}^L$.

By the inductive hypothesis, the recursive call on $(L, X_{F\cap L}\cup X_S, \mathbf{k}^L)$ gives $X'_L$ (note that it can be distinct from $X|_L$ but is still an EF-allocation),
     the same for $(R, X_{F\cap R}\cup X_S, \mathbf{k}^R)$ and $X'_R$.
Note that $X'_L$ and $X'_R$ agree, since they agree with $X_S$ and $L\cap R=S$.
By \Cref{thm:main}, combining $X'_L, X'_R$ and $X_S$ yields an EF-allocation that extends $X_F$.
It passes the run of the \texttt{IsEF} procedure.

If no envy-free allocation exists, the algorithm could not return anything other than \None, since no allocation could pass the check by \texttt{IsEF}.
The proof is complete.
\end{proof}

\paragraph{Running time analysis.}
We now analyze the running time of $\ConstructEF$, using the notion $p = |A|$, $K = (k_1+1) \cdot (k_2+1) \cdot (k_3+1)$.
To proceed, we clarify what does ``efficient'' means in the statement of \Cref{thm:main}.
First, $\Divide(A,X_S)$ runs in polynomial time, i.e.\  it takes $(pK)^{\Oh(1)}$ time.
Second, $\EnumerateSeparator$ runs in time $(pK)^{d\sqrt{p}}$ for some universal constant $d\ge 1$, given by the upper bound to the number of possible agent sets of size $\Oh(\sqrt{p})$ and possible partial allocations.

Our goal is to show that the overall running time is subexponential in the number of agents.

\paragraph{Recurrence for $\ConstructEF$.}
Let $T(p,\mathbf{k})$ denote the worst-case running time of $\ConstructEF$ on an instance with $p$ agents and item quantity vector $\mathbf{k}=(k_1,k_2,k_3)$.
We prove by induction that $T(p,\mathbf{k})\le (pK)^{c\sqrt{p}}$ for some constant $c>d$ to be specified later.
We naturally assume that $T(p',\mathbf{k}')\le T(p,\mathbf{k})$ for any $p'\le p$ and $\mathbf{k}'\le \mathbf{k}$.

If $p \le 1200$, the algorithm calls $\ExhaustiveSearch$.
Since this bound is just a constant, we have $T(p,\mathbf{k})=(pK)^{\Oh(1)}$ in this case.

Let us move on to $p>1200$.
For each candidate separator $X_S$ (at most $(pK)^{d\sqrt{p}}$ choices), the algorithm:
\begin{itemize}
    \item computes $(L,R)\gets \Divide(A, X_S)$ in polynomial time;
    \item iterates over all vectors $\mathbf{k}^L=(k_1^L,k_2^L,k_3^L)$ satisfying, for each $t\in [3]$, $k_t^S\le k_t^L \le k_t$.
\end{itemize}
The number of such vectors is at most $K$.
For each such choice, the algorithm makes two recursive calls, one on $L$ and one on $R$.

By the separator property, both sets $L$ and $R$ contain at most $\alpha p$ agents for some constant $\alpha<1$ ($C_2$ in \Cref{thm:main}). Thus,
$
|L| \le \alpha p
$
and
$
|R| \le \alpha p,
$
and the recursive call on each of them runs in at most $T(\alpha p, \mathbf{k})$ time. Then
\[
T(p,\mathbf{k}) \le
(pK)^{d\sqrt{p}}\cdot K\cdot (T(\alpha p,\mathbf{k})+T(\alpha p,\mathbf{k})).
\]
Equivalently, we may write
\[
T(p,\mathbf{k})
\le
(pK)^{d\sqrt{p}+2}\cdot T(\alpha p,\mathbf{k}).
\]
Using the bound on $T(\alpha p, \mathbf{k})$ by induction,
\[
T(p,\mathbf{k})
\le
(pK)^{d\sqrt{p}+2}\cdot (pK)^{c\sqrt{\alpha p}}=(pK)^{(d+c\sqrt{\alpha})\cdot \sqrt{p}+2}.
\]
It suffices to use that
\[(d+c\sqrt{\alpha})\cdot \sqrt{p}+2\le c\sqrt{p},\]
so choose any $c$ satisfying
\[2\le (c(1-\sqrt{\alpha})-d)\cdot \sqrt{1200}\]
to complete the proof of $T(p,\textbf{k})\le (pK)^{\Oh(\sqrt{p})}$.

\paragraph{Summing up.}
The discussion above leads to the following theorem, which is the main result of the paper.
\begin{theorem}\label{thm:algorithm}
    Given $n$ agents with non-degenerate additive valuations and $m$ items that come in three types, and a partial allocation of some items to some agents, an envy-free allocation extending the given partial allocation (if one exists) can be found in time $(nm)^{\Oh(\sqrt{n})}.$
\end{theorem}

\section{Enumerating Balanced Separators}\label{sec:enumerate}

The rest of the work is devoted to the proof of \Cref{thm:main}, that comes in several sections.
In this section, we introduce and study the \texttt{EnumerateSeparator} procedure.

As noted in the overview, any \EF allocation can be represented as a convex
polyhedron $P$ in $\mathbb{R}^3$, and the graph of $P$ is planar.
The procedure \EnumerateSeparator enumerates candidate $(S, X_S)$, where $X_S$ is a vector of separator agents with a partial allocation for it, such that the cycle $C$ induced by $X_S$ on $P$ is a valid planar cycle separator and $S$ is set of unique agents, presented in $X_S$. If an \EF allocation exists, the correct separator is guaranteed to be yielded at some point during enumeration.
The pseudocode is present in \Cref{alg:EnumerateSeparator}.

\begin{algorithm}
\small
\caption{\EnumerateSeparator\ (generator): enumerates candidate separators,
resuming from the last yielded value on each call.}
\label{alg:EnumerateSeparator}
\begin{algorithmic}[1]
\Require $A \subseteq [n]$ --- set of agents
\Require $\mathbf{k} = (k_1, k_2, k_3)$ --- number of items of each type
\Statex
\For{$\ell = 0,\ 1,\ \ldots,\ \lceil\sqrt{8|A|}\rceil$}
    \For{each $X_S \in \bigl( [|A| + 6] \times [k_1] \times [k_2] \times [k_3]\bigr)^{4\ell + 1}$}
        \If{$\CheckAllocation(X_S,\ \mathbf{k})$}
            \State $S\gets \text{unique agents, presented in $X_S$}$
            \State \Yield $(S, X_S)$
            \Comment{suspend here; next call resumes at this point}
        \EndIf
    \EndFor
\EndFor
\State \Yield \None
\end{algorithmic}
\end{algorithm}

We move on to the discussion and formal analysis of how the procedure works and theoretical facts it relies on.

\textbf{Bounding planes.} 
To handle our convex polyhedron properly, we require that it is bounded.
Note that any bundle point $x_i$ lies in a cube with extremal points $(0,0,0)$ and $(k_1,k_2,k_3)$.
We call the halfspace sequence $(\{x \geq 0\}, \{y \geq 0\}, \{z \geq 0\}, \{x\leq k_1\}, \{y \leq k_2\}, \{z \leq k_3\})$ the \emph{bounding planes vector}. These planes are used to complete set of agents planes, so the ``envy-freeness polyhedron'' is always bounded.
They can appear as a part of the separator; that is why $S$ is a suset of $A$ plus $6$ additional imaginary ``agents''.

\paragraph{Description of \CheckAllocation.}
The predicate $\CheckAllocation(X_S, \mathbf{k})$ verifies whether $X_S$ can serve as a valid separator in some \EF allocation. $X_S$ represents vector of planes, forming separator edges an vertices.

Note that here $X_S$ technically is a sequence of planes and points, and not just a partial allocation as we treated is in previous sections.
This is only a minor technical detail used by \texttt{EnumerateSeparator} and \texttt{Divide}.
$\CheckAllocation$ works as follows.

First, it checks that $X_S$ respects the item counts: $\sum_{i \in S\cap A}
(X_S)_{i} \leq \mathbf{k}$. If this fails, it returns \False.
We also check consistency here, if $i$ presented in $X_S$ in tuple $(i, (x_{i,1}, x_{i,2}, x_{i,3}))$ any other appearance in vector should happen with exactly same item allocation. Otherwise we return \False.

Next, for each agent with allocation $(i,  (x_{i,1}, x_{i,2}, x_{i,3})) = (X_S)_i$,
it constructs the hyperplane with normal $v_i$ (the valuation vector
of $i^\text{th}$ agent in $A$) passing through the point $(x_{i,1}, x_{i,2}, x_{i,3})$. If $i > |A|$ we take element with index $i-6$ in bounding planes vector and ignore the point itself.

For each $i$  we take hyperplane $H_i$ constructed in step above. Now we have sequence $V$, which contains $4\ell + 1$ hyperplanes. Let us split first  $4\ell$ hyperplanes into $\ell$ groups, each one formed of 4 consecutive elements.
For each
group, the four corresponding hyperplanes are intersected; if the
intersection is not a single point, the check fails.

The resulting $\ell$ intersection points are connected in their natural
order to form a candidate cycle $C$ (a closed polygonal chain in 3D space). If $C$ is not a simple polyline, the check fails.

Let $P'$ be the convex polyhedron formed by all hyperplanes $\{H_i\}_{i
\leq 4\ell + 1}$. The check then verifies that: every vertex of $C$ is a vertex of
$P'$; every edge of $C$ is either an edge of $P'$ or a  diagonal of a face of
$P'$; all $\ell$ hyperplanes appear as faces of $P'$; and each face of $P'$
contains at most one diagonal of $C$. Finally, it checks that the hyperplane
$H_{4\ell+1}$ (last plane in $V$) contains no vertex of
$C$. \label{def:separator_cycle}
$C$ is called separator cycle for $X_S$. If all checks pass, $\CheckAllocation$ returns \True.

\paragraph{Correctness.}

Correctness proof in this section relies on a classical theorem, established in~\cite{miller1986finding}.
\begin{theorem}[Miller]
\label{thm:miller}
Every maximal planar graph on $n$ vertices admits a simple cycle separator
of length at most $\sqrt{8n}$ such that each of the two regions of the
induced planar embedding contains at most $\frac{2}{3}n$ vertices.
\end{theorem}


Now we can formulate the correctness lemma, which further explains why point 2 of \Cref{thm:main} holds.
\begin{lemma}
\label{thm:enumerate}
If an \EF allocation $X$ for agent set $A$ with item counts
$\mathbf{k} = (k_1, k_2, k_3)$ exists, and $P$ is the bounded convex polyhedron
representing $X$, then \EnumerateSeparator\ yields $X_S$ such
that the induced separator cycle $C$ on $P$ satisfies:
\begin{enumerate}
    \item every vertex of $C$ is a vertex of $P$;
    \item every edge of $C$ is either an edge of $P$ or a face diagonal of
          $P$, with at most one diagonal per face;
    \item $|C| \leq \sqrt{8|A|}$;
    \item $C$ is a cycle separator of $P$ viewed as a planar graph;
    \item for the polyhedron $P'$ formed by the hyperplanes of $S$: an edge
          $e$ of $C$ is an edge of $P'$ if and only if it is an edge of $P$.
\end{enumerate}
\end{lemma}

\begin{proof}
Property~3 holds immediately since the outer loop runs only up to
$\lceil\sqrt{8|A|}\rceil$.

Assume an \EF allocation $X$ exists with polyhedron $P$. Triangulate each
face of $P$ by connecting one vertex to all others via diagonals, making $P$
a maximal planar graph. By Theorem~\ref{thm:miller}, there exists a simple
cycle separator $C'$ with $|C'| \leq \sqrt{8|A|}$. Since $C'$ is a simple
cycle, at most two diagonals per face can appear in it. For each face
containing two diagonals, replace both with the third diagonal completing
the triangle; this reduces $|C'|$ while preserving simplicity, so
property~2 holds.

Now consider the iteration at $\ell = |C'|$. For each vertex $v$ of $C'$,
let $e_1, e_2$ be its two incident edges in $C'$. For each $e \in \{e_1,
e_2\}$, if $e$ is a diagonal of face $F$ take $F$ twice; otherwise take the
two faces of $P$ containing $e$. This yields four faces whose hyperplanes
intersect at exactly $v$ (if fewer than four distinct faces arise, add any
face containing $v$). Record the item allocation $(t_1, t_2, t_3)$ and
agent index $i$ for each face and append $(i, (t_1, t_2, t_3))$ to the
candidate vector. Connecting these points in order recovers $C'$ exactly,
so $\CheckAllocation$ returns \True on this iteration and \EnumerateSeparator\
yields the corresponding $X_S$.

Properties~1, 2, and~4 follow by construction of $C'$. For property~5:
any edge of $C'$ that is an edge of $P$ is bounded by at least two faces
of $P'$, so it is an edge of $P'$. Conversely, if $e$ is a diagonal of
face $F$ in $P$, then only the hyperplane of $F$ contains $e$ in $P'$,
so $e$ cannot be an edge of $P'$.
\end{proof}

\paragraph{Running time.}
We move on to giving an upper bound to the running time. We follow the same $p, K$ notation as before.
\begin{lemma}
\label{lem:enum-complexity}
The total runtime of \EnumerateSeparator\ throughout all \Yield calls is
$(pK)^{\Oh(\sqrt{p})}$.
\end{lemma}
\begin{proof}
At iteration $\ell$, the number of candidate  $X_S$ is at most
$(k_1 \cdot k_2 \cdot k_3 \cdot (p + 6))^{4\ell + 1}$. Each call to $\CheckAllocation$
runs in $p^{\Oh(1)}$ time, since each condition involves only
elementary geometric operations (hyperplane intersection, planarity checks)
on $\ell\le p$ objects. Summing over all iterations:
\[
    \sum_{\ell=0}^{\lceil\sqrt{8n}\rceil}
    (K \cdot (p + 6))^{4\ell + 1} \cdot p^{\Oh(1)}
    \leq
     (Kp)^{4\sqrt{8p}+3} \cdot p^{\Oh(1)}.
\]
\end{proof}

\noindent To summarize, we note that point~\ref{prop:size} of \Cref{thm:main} follows from \Cref{thm:enumerate}; and \EnumerateSeparator runs efficiently (in subexponential time).

\section{Dividing the Agents}
\label{sec:divide}

In this section we describe the \Divide procedure, which takes the separator $X_S$ (as a sequence of hyperplanes) produced by \EnumerateSeparator and partitions the remaining agents into two balanced groups. The construction relies on the notion of the \emph{spherical image} of a convex polyhedron and certain of its structural properties.

\begin{definition}
Let $Q$ be a convex surface and let $\mathbb{S}$ denote the unit sphere in $\mathbb{R}^t$ for some $t \ge 1$.
For a set $M$ of points on $Q$, draw all support planes to $Q$ at each point of $M$ and translate their outward unit normals to the origin.
The set of endpoints of these normals on $\mathbb{S}$ is called the
\emph{spherical image} of $M$, denoted $s(M)$.
\end{definition}
We will use the following property of spherical images.

\begin{theorem}[Alexandrov~\cite{alexandrov2005spherical}.]
\label{thm:spherical}
The spherical image $s(P)$ of a convex polyhedron $P$ with $p$ vertices is a decomposition of
$\mathbb{S}$ into convex spherical polygons $s(v_1), \ldots, s(v_p)$
with pairwise disjoint interiors, where:
\begin{itemize}
    \item each polygon $s(v_i)$ is the spherical image of vertex $v_i$ of $P$;
    \item each common edge of $s(v_i)$ and $s(v_j)$ is the spherical image
          $s(e_{ij})$ of the interior of edge $e_{ij} = v_iv_j$ of $P$;
    \item each common vertex of polygons $s(v_{i_1}), \ldots, s(v_{i_t})$
          is the spherical image $s(f)$ of the interior of the face $f$ of $P$
          whose vertices are $v_{i_1}, \ldots, v_{i_t}$.
\end{itemize}
\end{theorem}

\paragraph{High-level idea.} The high-level idea of \texttt{Divide} is to take the output of \texttt{EnumerateSeparator} in the form of a convex polyhedron and a closed polyline, construct the spherical image $\Gamma$ of the polyline and split the agents according to their valuation point placement relative to $\Gamma$ as a spherical curve.

\paragraph{Construction of $\Gamma$.}

Recall that \EnumerateSeparator yields a set of hyperplanes whose intersection
defines a convex polyhedron $P'$, and produces a cycle on the graph of
$P'$ that serves as a separator. We construct a corresponding curve $\Gamma$
on $\mathbb{S}$ as follows.

\begin{enumerate}
    \item Let $P'$ be the convex polyhedron formed by the hyperplanes fixed
          in \EnumerateSeparator.
    \item Compute the spherical image $s(P')$. By Theorem~\ref{thm:spherical},
          the image $s(v)$ of each vertex $v$ of $P'$ is a convex spherical
          polygon, and the image $s(e)$ of each edge $e$ of the separator cycle
          is either a geodesic arc (if $e$ is an edge of $P'$) or a single
          point (if $e$ is a face diagonal, in which case $s(e) = s(f)$ for
          the face $f$ containing $e$).
    \item For each edge $e$ of the separator cycle : if $s(e)$ is an arc,
          take its midpoint; if $s(e)$ is a point, take that point. Connect
          the resulting points in the cyclic order inherited from the separator
          cycle to obtain the polyline $\Gamma$ (see Figure~\ref{fig:fig2}).
    \item The curve $\Gamma$ divides $\mathbb{S}$ into two open regions.
          For each agent $i \in A \setminus S$, place
          its preference vector $u_i$ (the $i$-th row of $U$) on
          $\mathbb{S}$ and assign agent $i$ to $L$ or $R$ according to
          which region contains $u_i$.
\end{enumerate}

\begin{figure}[h!]
    \centering
    \usetikzlibrary{arrows}
\definecolor{qqqqff}{rgb}{0,0,1}
\definecolor{qqwwtt}{rgb}{0,0.4,0.2}
\definecolor{ttqqqq}{rgb}{0.2,0,0}
\definecolor{zzttqq}{rgb}{0.6,0.2,0}
\definecolor{qqwuqq}{rgb}{0,0.39,0}

\begin{tikzpicture}[line cap=round, line join=round, >=triangle 45, x=1.0cm, y=1.0cm]

\draw [thick] (-10.17, 10.46) -- node[left, xshift=-2pt] {$e_1$} (-11.28, 7.20);
\draw [thick, color=qqwuqq] (-11.28, 7.20) -- node[left, xshift=-2pt] {$e_2$} (-9.84, 3.67);
\draw [thick, color=zzttqq] (-9.84, 3.67) -- node[left, xshift=-2pt] {$e_3$} (-10.75, 0.26);

\foreach \x/\y in {-10.17/10.46, -11.28/7.20, -9.84/3.67, -10.75/0.26} {
    \fill [black] (\x,\y) circle (1.5pt);
}

\begin{scope}[xshift=-2.1cm]

    \draw [thick] (-4.24, 10.58) -- node[above] {$s(e_1)$} (0.30, 10.60);
    \draw [thick, color=qqwuqq] (-5.53, 6.86) -- node[above, xshift=15pt] {$s(e_2)$} (-0.23, 5.31);

    \draw [thick, color=qqqqff] (-1.97, 10.59) -- (-2.88, 6.09) node[left=5pt,yshift=-4pt] {$\Gamma$} -- (-1.62, 0.63);

    \node[color=zzttqq, right=4pt] at (-1.62, 0.63) {$s(e_3)$};

    \foreach \x/\y in {-4.24/10.58, 0.30/10.60, -5.53/6.86, -0.23/5.31} {
        \fill [black] (\x,\y) circle (1.5pt);
    }

    \fill [ttqqqq] (-1.97, 10.59) circle (1.5pt);
    \fill [qqwwtt] (-2.88, 6.09) circle (1.5pt);
    \fill [zzttqq] (-1.62, 0.63) circle (1.5pt);

\end{scope}

\end{tikzpicture}
    \caption{The part of separator cycle on the polyhedron $P'$ and its spherical
             image with $\Gamma$. $e_i$ are the edges of separator, $s(e_i)$ are their corresponding spherical images.}
    \label{fig:fig2}
\end{figure}
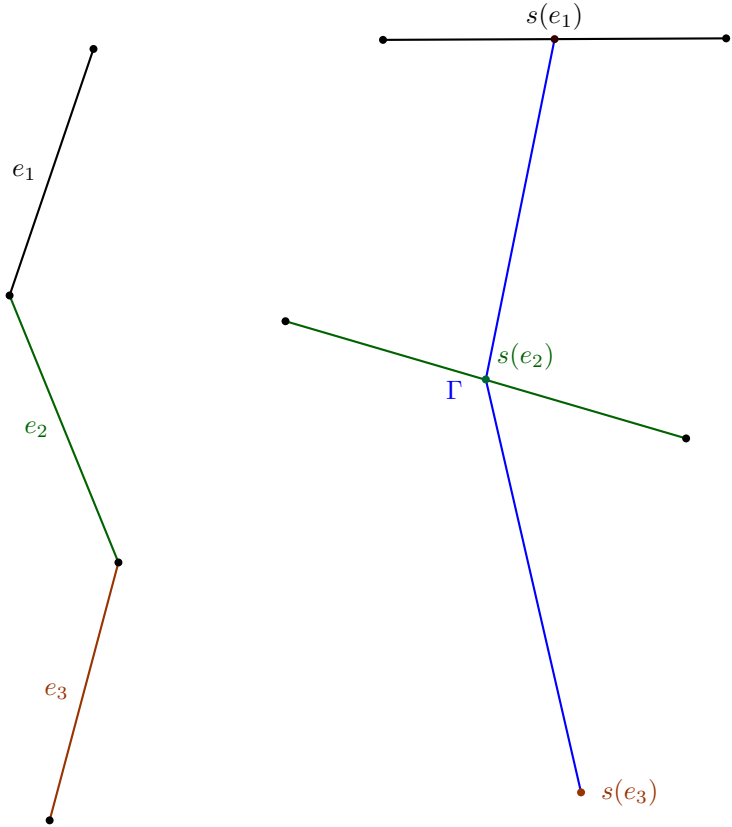

\begin{lemma}
\label{lem:regular}
The polyline $\Gamma$ is a simple closed curve on $\mathbb{S}$.
\end{lemma}

\begin{proof}[Proof of \Cref{lem:regular}]
We show that $\Gamma$ does not self-intersect and that every vertex of
$\Gamma$ has degree exactly~2.

\textbf{\textit{No self-intersection in edge interiors.}}
By \Cref{thm:spherical}, the interior of each spherical polygon
$s(v_i)$ contains no edges of $s(P')$.
Every edge of $\Gamma$ lies in the interior of some $s(v_i)$,
so no two edges of $\Gamma$ can cross in their interiors.
Moreover, two edges of $\Gamma$ cannot both lie inside the same polygon
$s(v_i)$: each edge of $\Gamma$ connects images of edges of the separator
cycle adjacent to a single vertex, so two edges inside the same $s(v_i)$
would imply that the corresponding vertex of the separator cycle has degree
at least~4, contradicting the fact that it is a simple cycle.

\textbf{\textit{Every vertex has degree~2.}}
Vertices of $\Gamma$ that are midpoints of arcs $s(e)$ are each visited
exactly once: visiting such a vertex twice would require the same edge
$e$ of $P'$ to appear twice in the separator cycle, which is impossible.
Vertices of the form $s(f)$ for a face diagonal are also each visited
at most once, because \EnumerateSeparator guarantees that every face of $P'$
contains at most one diagonal of the separator cycle; each visit corresponds
to traversing that diagonal, so it can occur at most once.
\end{proof}

\begin{lemma}
\label{lem:nocross}
Let $P$ be any convex polyhedron such that
\emph{(i)} the faces of $P$ include all faces of $P'$, and
\emph{(ii)} the separator cycle lies on the surface of $P$, with its vertices
being vertices of $P$ and its edges being edges or face diagonals of $P$,
consistently with their status in $P'$.
Then no edge of $s(P)$ crosses $\Gamma$, except for the images $s(e)$ of
the edges $e$ of the cycle.
\end{lemma}

\begin{proof}[Proof of \Cref{lem:nocross}]
    
By Lemma~\ref{lem:regular}, no edge of $s(P)$ crosses an edge of $\Gamma$
in an interior point, since every such interior point lies in the interior
of some $s(v_i)$, which contains no edges of $s(P)$. It remains to check
whether an edge of $s(P)$ passes through a vertex of $\Gamma$. By
Theorem~\ref{thm:spherical}, the only vertices of $\Gamma$ that can lie on
an edge of $s(P)$ are midpoints of arcs $s(e)$ for edges $e$ of $P'$.
Such a vertex lies on an edge of $s(P)$ only if it is the image of an
edge of $P$, which by construction means that edge belongs to the separator
cycle.
\end{proof}

\paragraph{Correctness.}

We show that any convex polyhedron (the complete solution) that extends $P'$ (partial allocation given by the separator) agrees with the division of the agents.
First, split algorithm is correct, since Lemma~\ref{lem:regular} holds, so $\Gamma$ divides $\mathbb{S}$ into two open regions by Jordan curve theorem. Next, by Lemma~\ref{lem:nocross} we can see that if we build $\Gamma$ on spherical image of any complete solution, we will get exactly same split, since $\Gamma$ and spherical images of agents planes does not change.

\paragraph{Balancedness.}

By Lemma~\ref{lem:regular} and the Jordan curve theorem, $\Gamma$ divides
$\mathbb{S}$ into exactly two open regions. By
Theorem~\ref{thm:spherical}, each face of $s(P)$ corresponds to a vertex
of $P$ and hence to an agent in $A$. Since the separator
cycle is produced by Miller's theorem (Theorem~\ref{thm:miller}), both
regions contain at least $\frac{1}{3}|A|$ faces of $s(P)$.

By Euler's formula $V - E + F = 2$, together with $E \geq \frac{3}{2}F$
(every face has at least 3 edges, every edge borders exactly 2 faces), we
obtain $V \geq \frac{1}{2}F$. Hence each region contains at least
$
    \frac{1}{2} \cdot \frac{1}{3}|A|
    = \frac{|A|}{6}$ free agents (agents in $A\setminus S$), giving
$
    |L \setminus S|,\; |R \setminus S| \;\geq\; \frac{|A|}{6} - |S| \geq \frac{|A|}{12},
$
so $
    |L|, |R| \leq\frac{11}{12}|A|.$
This unequality establishes the balancedness constant $C_2 = \frac{11}{12}$ in \Cref{thm:main}. $\frac{|A|}{6} - |S| \geq \frac{|A|}{12}$ is satisfied since if we run \Divide, number of agents is over 1200 and $|S| \leq \sqrt{8n} \leq \frac{1}{12}n$. The point~\ref{prop:balance} of \Cref{thm:main}
 follows immediately from the above reasoning.

\section{Merging Solutions from Two Parts}
\label{sec:merge}

In this section we prove that once \Divide has partitioned the agents into two groups and independent \EF allocations have been found for each part, their union is an \EF allocation.

\paragraph{Setup.}
Let $L$ and $R$ be the two parts produced by \Divide, with $L \cap R = S$ where $S$ is the separator set. Let $X_L$ and $X_R$ be \EF allocations for $L$ and $R$ respectively, where in each case the agents in $S$ are treated as fixed with allocation $X_S$. Define $X = X_L \cup X_R$.

\paragraph{Main argument.}
We show that $X$ is envy-free for all of $A$. Suppose for contradiction that $X$ is not \EF. Then there exist agents $a_L \in L \setminus S$ and $a_R \in R \setminus S$ such that $a_L$ envies $a_R$ under $X$.

Consider the convex polyhedron $P'$ formed by the hyperplanes of $S$ together with the hyperplanes $H_L$ and $H_R$ corresponding to $a_L$ and $a_R$ respectively. Assuming for now that $P'$ is non-degenerate (no two hyperplanes coincide), we can form its spherical image $s(P')$.

By the construction in \Divide, the curve $\Gamma = s(C)$ is a simple
closed curve on $\mathbb{S}$ that separates $s(H_L)$ from $s(H_R)$.
By Lemma~\ref{lem:nocross}, no edge of $s(P')$ crosses $\Gamma$ except for images of edges of $C$, so $a_L$ and $a_R$ are not adjacent in $P'$. It follows that every face of $P'$ adjacent to $H_R$ is a face of some agent in $S$, and each such face lies on the correct side of $H_L$. Therefore every point in the face corresponding to $a_R$ is not envied by $a_L$, contradicting the assumption.

It remains to show that $P'$ is always non-degenerate, i.e.\ that $H_L$ and $H_R$ cannot simultaneously make $P'$ singular.
Recall that \EnumerateSeparator yields, as its last element, a hyperplane
$H_F$ that contains no vertex of the separator cycle $C$
(last condition of \CheckAllocation). This hyperplane $H_F$ lies strictly on one side of $C$: for one of the two parts, say $L$, it lies on the same side as the non-separator faces of $X_L$; for the other part $R$, it lies on the opposite side.

Consequently, in the allocation $X_L$ all faces lie on the same side of $C$ as $H_F$, while in $X_R$ all faces lie on the opposite side. When we form the polyhedron determined by $S$ alone, the faces $H_L$ and $H_R$
therefore lie on opposite sides of $C$. Adding both $H_L$ and $H_R$ simultaneously cannot make $P'$ singular, since any singularity would require them to be on the same side, which is impossible.

We proved that $X = X_L \cup X_R$ is an \EF allocation for all of $A$, so the last point (point~\ref{prop:merge}) of \Cref{thm:main} holds. This completes the proof of Theorem~\ref{thm:main}.

\section{Conclusion}

In this work, we presented a subexponential algorithm for finding an envy-free allocation of indivisible items among $n$ agents when items come in three distinct types. The algorithm handles goods, chores, and mixed instances uniformly: agents may assign positive, negative, or mixed valuations to item types, and the \EF guarantee holds in all cases. Correctness is established via a geometric representation of \EF allocations as bounded convex polyhedra in $\mathbb{R}^3$, combined with a divide-and-conquer decomposition based on
Miller's planar cycle-separator theorem. We also extended the result to the setting with fixed agents, where a subset of allocations is given as part of
the input.

The present paper closes the open case $t = 3$, reducing the complexity from brute-force exponential to subexponential, and fits into a broader complexity
picture parameterised by the number of types: polynomial at $t = 2$, subexponential at $t = 3$, and $2^{\Omega(m)}$-hard under ETH for $t = \Omega(m)$.

\paragraph{Open problems and future directions.}
The most natural open question is whether our result generalises to an arbitrary number of item types. We conjecture that for any fixed number of types $t$, there exists an algorithm with time complexity
\[
    (nm)^{\Oh\left(n^{1 - \frac{1}{t-1}}\right)}.
\]
Proving this conjecture for general $t$ would likely require a higher-dimensional analogue of the planar separator technique, possibly drawing on separator theorems for $d$-dimensional polytopes or higher-genus surfaces.

Beyond the complexity question, several directions remain open.
\begin{itemize}
    \item \textbf{Tightness.} Is the exponent $\sqrt{n}$ optimal for $t = 3$,
    or can it be improved? A matching lower bound under standard complexity
    assumptions would strengthen the result.

    \item \textbf{Approximate fairness.} Can the geometric approach be adapted
    to compute EF1 or EFX allocations more efficiently than existing
    combinatorial methods, in particular for chores and mixed instances where
    EFX existence is still open?

    \item \textbf{Item-specific perturbations.} What happens when agents have
    valuations that are not purely type-based, but include item-specific
    additive perturbations? The convex-polyhedron representation breaks down in
    this setting, and new geometric tools would be needed.
\end{itemize}

\section{Acknowledgments}
This work was supported by the Ministry of Economic Development of the Russian Federation (IGK 000000C313925P4C0002), agreement №139-15-2025-010

\bibliographystyle{plain}
\bibliography{references}

@inproceedings{10.5555/3463952.3463988,
title = {High-Multiplicity Fair Allocation Made More Practical},
author = {Robert Bredereck and Aleksander Figiel and Andrzej Kaczmarczyk and Dušan Knop and Rolf Niedermeier},
year = {2021},
date = {2021-01-01},
booktitle = {Proceedings of the 20th International Conference on Autonomous Agents and Multiagent Systems (AAMAS '21)},
pages = {260–268},
publisher = {ACM},
pubstate = {published},
tppubtype = {inproceedings}
}

@article{bouveret2008efficiency,
  title={Efficiency and envy-freeness in fair division of indivisible goods: Logical representation and complexity},
  author={Bouveret, Sylvain and Lang, J{\'e}r{\^o}me},
  journal={Journal of Artificial Intelligence Research},
  volume={32},
  pages={525--564},
  year={2008}
}

@inproceedings{bliem2016complexity,
  title={Complexity of Efficient and Envy-Free Resource Allocation: Few Agents, Resources, or Utility Levels.},
  author={Bliem, Bernhard and Bredereck, Robert and Niedermeier, Rolf},
  booktitle={Proceedings of the 25th International Joint Conference on
               Artificial Intelligence (IJCAI)},
  pages={102--108},
  year={2016}
}

@article{miller1986finding,
  title={Finding small simple cycle separators for 2-connected planar graphs},
  author={Miller, Gary L.},
  journal={Journal of Computer and System Sciences},
  volume={32},
  pages={265--279},
  year={1986}
}

@inbook{alexandrov2005spherical,
  title={Convex Polyhedra},
  author={Alexandrov, Alexander D.},
  chapter={Spherical Image},
  publisher={Springer},
  address={Berlin, Heidelberg},
  year={2005}
}

@book{moulin2004fair,
  author    = {Herv{\'e} Moulin},
  title     = {Fair Division and Collective Welfare},
  publisher = {MIT Press},
  address   = {Cambridge, MA},
  year      = {2004}
}

@book{brandt2016handbook,
  editor    = {Felix Brandt and Vincent Conitzer and Ulle Endriss
               and J{\'e}r{\^o}me Lang and Ariel D. Procaccia},
  title     = {Handbook of Computational Social Choice},
  publisher = {Cambridge University Press},
  year      = {2016}
}

@article{steinhaus1948problem,
  author  = {Hugo Steinhaus},
  title   = {The problem of fair division},
  journal = {Econometrica},
  volume  = {16},
  pages   = {101--104},
  year    = {1948}
}

@inproceedings{lipton2004approximately,
  author    = {Richard J. Lipton and Evangelos Markakis
               and Elchanan Mossel and Amin Saberi},
  title     = {On approximately fair allocations of indivisible goods},
  booktitle = {Proceedings of the 5th ACM Conference on Electronic Commerce
               (EC)},
  pages     = {125--131},
  year      = {2004}
}

@article{budish2011combinatorial,
  author  = {Eric Budish},
  title   = {The combinatorial assignment problem: Approximate competitive
             equilibrium from equal incomes},
  journal = {Journal of Political Economy},
  volume  = {119},
  number  = {6},
  pages   = {1061--1103},
  year    = {2011}
}

@article{caragiannis2019unreasonable,
  author  = {Ioannis Caragiannis and David Kurokawa and Herv{\'e} Moulin
             and Ariel D. Procaccia and Nisarg Shah and Junxing Wang},
  title   = {The unreasonable fairness of maximum {Nash} welfare},
  journal = {ACM Transactions on Economics and Computation},
  volume  = {7},
  number  = {3},
  pages   = {12:1--12:32},
  year    = {2019}
}

@inproceedings{aziz2015fair,
  author    = {Haris Aziz and Simon Mackenzie},
  title     = {A discrete and bounded envy-free cake cutting protocol for
               any number of agents},
  booktitle = {Proceedings of the 57th IEEE Symposium on Foundations of
               Computer Science (FOCS)},
  pages     = {416--427},
  year      = {2016}
}

@book{cygan2015parameterized,
  author    = {Marek Cygan and Fedor V. Fomin and {\L}ukasz Kowalik
               and Daniel Lokshtanov and D{\'a}niel Marx
               and Marcin Pilipczuk and Micha{\l} Pilipczuk
               and Saket Saurabh},
  title     = {Parameterized Algorithms},
  publisher = {Springer},
  address   = {Cham},
  year      = {2015}
}

@inproceedings{mackin2016allocating,
  author    = {Erika Mackin and Lirong Xia},
  title     = {Allocating indivisible items in categorized domains},
  booktitle = {Proceedings of the 25 International Joint Conference on Artificial Intelligence (IJCAI)},
  year      = {2016}
}

@article{nguyen2014mnw,
author = {Nguyen, Trung Thanh and Rothe, Jörg},
year = {2014},
month = {12},
pages = {54-68},
title = {Minimizing envy and maximizing average Nash social welfare in the allocation of indivisible goods},
volume = {179},
journal = {Discrete Applied Mathematics},
doi = {10.1016/j.dam.2014.09.010}
}

@inproceedings{bhaskar2021chores,
  author    = {Umang Bhaskar and A.~R. Sricharan and Rohit Vaish},
  title     = {On approximate envy-freeness for indivisible chores and
               mixed resources},
  booktitle = {Approximation, Randomization, and Combinatorial Optimization
               (APPROX/RANDOM)},
  year      = {2021}
}

@inproceedings{aziz2021mixedmanna,
  author    = {Haris Aziz and Ioannis Caragiannis and Ayumi Igarashi
               and Toby Walsh},
  title     = {Fair allocation of indivisible goods and chores},
  booktitle = {Proceedings of the 28th International Joint Conference on
               Artificial Intelligence (IJCAI)},
  pages     = {53--59},
  year      = {2019}
}

@article{bogomolnaia2017dividing,
author = {Bogomolnaia, Anna and Moulin, Hervé and Sandomirskiy, Fedor and Yanovskaya, Elena},
year = {2016},
month = {10},
pages = {},
title = {Dividing Goods AND Bads Under Additive Utilities},
journal = {SSRN Electronic Journal},
doi = {10.2139/ssrn.2826474}
}

@article{EFX_PO, 
title={EFX and PO Allocation Exists for Two Types of Goods}, 
volume={40}, url={https://ojs.aaai.org/index.php/AAAI/article/view/38723}, 
DOI={10.1609/aaai.v40i20.38723}, 
number={20}, 
journal={Proceedings of the AAAI Conference on Artificial Intelligence}, 
author={Davidiuk, Vladimir and Dementiev, Yuriy and Ignatiev, Artur and Sagunov, Danil}, year={2026}, 
month={Mar.}, 
pages={16795–16802} }

\end{document}